\documentclass[11pt,letterpaper]{article}

\usepackage{mathpazo}
\usepackage{dsfont}
\usepackage{amssymb}
\usepackage{amsmath}
\usepackage{latexsym}
\usepackage{amsfonts}

\usepackage{hyperref}
\hypersetup{pdfpagemode=UseNone}

\usepackage{amsthm}

\newtheorem{theorem}{Theorem}
\newtheorem{lemma}[theorem]{Lemma}

\theoremstyle{definition}
\newtheorem{definition}[theorem]{Definition}

\newcommand{\tinyspace}{\mspace{1mu}}
\newcommand{\microspace}{\mspace{0.5mu}}
\newcommand{\op}[1]{\operatorname{#1}}
\newcommand{\tr}{\operatorname{Tr}}

\renewcommand{\vec}{\operatorname{vec}}
\newcommand{\fid}{\operatorname{F}}

\renewcommand{\t}{{\scriptscriptstyle\mathsf{T}}}

\newcommand{\abs}[1]{\lvert #1 \rvert}

\newcommand{\Bigabs}[1]{\Bigl\lvert #1 \Bigr\rvert}

\newcommand{\ip}[2]{\langle #1 , #2\rangle}
\newcommand{\bigip}[2]{\bigl\langle #1, #2 \bigr\rangle}
\newcommand{\Bigip}[2]{\Bigl\langle #1, #2 \Bigr\rangle}

\newcommand{\norm}[1]{\lVert\tinyspace#1\tinyspace\rVert}
\newcommand{\bignorm}[1]{\bigl\lVert\tinyspace #1 \tinyspace\bigr\rVert}
\newcommand{\Bignorm}[1]{\Bigl\lVert\tinyspace #1 \tinyspace\Bigr\rVert}

\newcommand{\triplenorm}[1]{
  |\!\microspace|\!\microspace| #1 
  |\!\microspace|\!\microspace|}

\newcommand{\setft}[1]{\mathrm{#1}}
\newcommand{\Lin}{\setft{L}}
\newcommand{\Density}{\setft{D}}
\newcommand{\Herm}{\setft{Herm}}
\newcommand{\Pos}{\setft{Pos}}
\newcommand{\Pd}{\setft{Pd}}
\newcommand{\Unitary}{\setft{U}}
\newcommand{\Trans}{\setft{T}}

\def\complex{\mathbb{C}}

\def\I{\mathds{1}}

\def\X{\mathcal{X}}
\def\Y{\mathcal{Y}}
\def\Z{\mathcal{Z}}
\def\W{\mathcal{W}}
\def\P{\mathcal{P}}
\def\D{\mathcal{D}}

\newenvironment{mylist}[1]{\begin{list}{}{
	\setlength{\leftmargin}{#1}
	\setlength{\rightmargin}{0mm}
	\setlength{\labelsep}{2mm}
	\setlength{\labelwidth}{8mm}
	\setlength{\itemsep}{0mm}}}
	{\end{list}}

\begin{document}

\title{\bf Simpler semidefinite programs for \\ completely bounded norms}
\author{
  {\large John Watrous}\\[2mm]
  {\it Institute for Quantum Computing and School of Computer Science}\\[1mm]
  {\it University of Waterloo}}
\date{August 2, 2012}
\maketitle

\begin{abstract}
  The completely bounded trace and spectral norms, for
  finite-dimensional spaces, are known to be efficiently expressible
  by semidefinite programs (J. Watrous, \emph{Theory of Computing} 5:
  11, 2009).
  This paper presents two new, and arguably much simpler, semidefinite
  programming formulations of these norms.
\end{abstract}

\section{Introduction and preliminary discussion}
\label{sec:introduction}

In the theory of quantum information, \emph{quantum states} are
represented by density operators acting on finite-dimensional complex
vector spaces, while \emph{quantum channels} are represented by linear
mappings that transform one density operator into another
\cite{NielsenC00,KitaevSV02}.
Various concepts connected with mappings of this form, meaning ones
that map linear operators to linear operators (or, equivalently,
that map matrices to matrices), are important in the study of quantum
information for this and other reasons.
Linear mappings of this form are also important in the study of
operator algebras \cite{Paulsen02}.

This paper is concerned specifically with the 
\emph{completely bounded trace and spectral norms},
defined for linear mappings of the form just described.
It is intended as a follow-up paper to \cite{Watrous09}, which
demonstrated that these norms can be efficiently expressed and
computed through the use of semidefinite programming.
Two new semidefinite programming formulations of these norms will be
presented, both of which are simpler than the formulations given in
the previous paper.

A further discussion of the completely bounded trace and spectral
norms can be found in \cite{Watrous09}.
That discussion will not be repeated here---instead, we will proceed
directly to the technical content of the paper, beginning with a short
summary of the notation and basic concepts that are to be assumed.

\subsubsection*{Linear algebra basics}

For a complex vector space of the form $\X = \complex^n$ 
and vectors $u = (u_1,\ldots,u_n)$ and $v = (v_1,\ldots,v_n)$ in $\X$,
we define the inner product
\[
\ip{u}{v} = \sum_{j = 1}^n \overline{u_j} v_j
\]
as well as the Euclidean norm
\[
\norm{u} = \sqrt{\langle u, u\rangle}.
\]
For each $j\in\{1,\ldots,n\}$, the vector $e_j\in\X$ is defined to be
the vector having a 1 in entry $j$ and 0 for all other entries.

Given two complex vector spaces $\X = \complex^n$ and 
$\Y = \complex^m$, we denote the space of all linear mappings (or
\emph{operators}) of the form $A:\X\rightarrow\Y$ as $\Lin(\X,\Y)$,
and identify this space with the collection of all $m\times n$ complex
matrices in the usual way.
For each pair of indices $(i,j)$ we write $E_{i,j}$ to denote the
operator whose matrix representation has a 1 in entry $(i,j)$ and
zeroes in all other entries.
The notation $\Lin(\X)$ is shorthand for $\Lin(\X,\X)$, and the
identity operator on $\X$, which is an element of $\Lin(\X)$,
is denoted $\I_{\X}$.
(The notation $\I$ is sometimes used in place of $\I_{\X}$ when it is
clear that we are referring to the identity operator on $\X$.)

For each operator $A\in\Lin(\X,\Y)$, one defines
$A^{\ast}\in\Lin(\Y,\X)$ to be the unique operator satisfying
$\ip{v}{Au} = \ip{A^{\ast}v}{u}$ for all $u\in\X$ and $v\in\Y$.
As a matrix, $A^{\ast}$ is obtained by taking the conjugate transpose
of the matrix associated with~$A$.
An inner product on $\Lin(\X,\Y)$ is defined as 
\[
\ip{A}{B} = \tr(A^{\ast} B)
\]
for all $A,B\in\Lin(\X,\Y)$.
By identifying a given vector $u\in\X$ with the linear mapping 
$\alpha \mapsto \alpha u$, which is an element of $\Lin(\complex,\X)$,
the mapping $u^{\ast}\in\Lin(\X,\complex)$ is defined.
More explicitly, $u^{\ast}$ is the mapping that satisfies
$u^{\ast} v = \ip{u}{v}$ for all $v\in\X$.

An operator $X\in\Lin(\X)$ is \emph{Hermitian} if $X = X^{\ast}$, and
the set of such operators is denoted $\Herm(\X)$.
An operator $X\in\Lin(\X)$ is \emph{positive semidefinite} if it is
Hermitian and all of its eigenvalues are nonnegative.
The set of such operators is denoted $\Pos(\X)$.
The notation $X\geq 0$ also indicates that $X$ is positive
semidefinite, and more generally the notations $X\leq Y$ and $Y\geq X$
indicate that $Y - X\geq 0$ for Hermitian operators $X$ and $Y$.
An operator $X\in\Lin(\X)$ is \emph{positive definite} if it is
both positive semidefinite and invertible.
Equivalently, $X$ is positive definite if it is Hermitian and all of
its eigenvalues are positive.
The set of such operators is denoted $\Pd(\X)$.
The notation $X>0$ also indicates that $X$ is positive definite, and
the notations $X<Y$ and $Y>X$ indicate that $Y - X> 0$ for
Hermitian operators $X$ and $Y$.
An operator $\rho\in\Lin(\X)$ is a \emph{density operator} if it is
both positive semidefinite and has trace equal to~1, and the set of
such operators is denoted $\Density(\X)$.
Finally, an operator $U\in\Lin(\X)$ is \emph{unitary} if 
$U^{\ast} U = \I_{\X}$, and the set of such operators is denoted
$\Unitary(\X)$.

For $\X = \complex^n$ and $\Y = \complex^m$, the space of all linear
mappings of the form $\Phi:\Lin(\X) \rightarrow \Lin(\Y)$ is denoted
$\Trans(\X,\Y)$.
For each $\Phi\in\Trans(\X,\Y)$, the mapping
$\Phi^{\ast}\in\Trans(\Y,\X)$ is the unique mapping for which the
equation
\[
\ip{Y}{\Phi(X)} = \ip{\Phi^{\ast}(Y)}{X}
\]
holds for all $X\in\Lin(\X)$ and $Y\in\Lin(\Y)$.
A mapping $\Phi\in\Trans(\X,\Y)$ is \emph{Hermiticity preserving} if it
holds that $\Phi(X) \in\Herm(\Y)$ for all choices of $X\in\Herm(\X)$,
\emph{positive} if it holds that $\Phi(X) \in \Pos(\Y)$ for all
$X\in\Pos(\X)$, and \emph{completely positive} if
$\Phi\otimes\I_{\Lin(\complex^k)}$ is positive for all $k\geq 1$.

\subsubsection*{Norms and fidelity}

For $\X = \complex^n$, $\Y = \complex^m$, and any operator
$A\in\Lin(\X,\Y)$, one defines the 
\emph{trace norm}, \emph{Frobenius norm}, and \emph{spectral norm} as
\[
\norm{A}_1 = \tr\sqrt{A^{\ast} A}\:,
\quad
\norm{A}_2 = \sqrt{\ip{A}{A}}\:,
\quad\text{and}\quad
\norm{A}_{\infty} =
\max \bigl\{\norm{Au}\,:\,u\in\X,\,\norm{u} \leq 1\bigr\},
\]
respectively.
These norms correspond precisely to the 1-norm, 2-norm, and
$\infty$-norm of the vector of singular values of $A$.
All three of these norms are \emph{unitarily invariant}, meaning that
\[
\norm{U A V}_1 = \norm{A}_1,
\qquad
\norm{U A V}_2 = \norm{A}_2,
\qquad\text{and}\qquad
\norm{U A V}_\infty = \norm{A}_\infty
\]
for every operator $A\in\Lin(\X,\Y)$ and every choice of unitary
operators $U\in\Unitary(\Y)$ and $V\in\Unitary(\X)$.
For every operator $A\in\Lin(\X,\Y)$ it holds that
\[
\norm{A}_{\infty} \leq \norm{A}_2 \leq \norm{A}_1.
\]
The trace and spectral norms are dual to one another, meaning
\begin{align*}
\norm{A}_1 & = \max\bigl\{\abs{\ip{B}{A}}\,:\,\norm{B}_{\infty} \leq 1\bigr\},\\
\norm{A}_{\infty} & = \max\bigl\{\abs{\ip{B}{A}}\,:\,\norm{B}_1 \leq 1\bigr\},
\end{align*}
for all $A\in\Lin(\X,\Y)$, and with $B$ ranging over operators within
the same space, while the Frobenius norm is self-dual.

For each $\Phi\in\Trans(\X,\Y)$, one defines the \emph{induced} trace
and spectral norms as
\begin{align*}
\norm{\Phi}_1 
& = \max\bigl\{\norm{\Phi(X)}_1 \,:\, X\in\Lin(\X),\,\;
\norm{X}_1\leq 1\bigr\},\\
\norm{\Phi}_\infty
& = \max\bigl\{\norm{\Phi(X)}_\infty \,:\, X\in\Lin(\X),\,\;
\norm{X}_\infty\leq 1\bigr\},
\end{align*}
as well as \emph{completely bounded} variants of these norms:
\begin{align*}
\triplenorm{\Phi}_1 
& = \sup_{k\geq 1}\,
\bignorm{\Phi\otimes\I_{\Lin(\complex^k)}}_1
= \bignorm{\Phi\otimes\I_{\Lin(\X)}}_1,\\
\triplenorm{\Phi}_\infty
& = \sup_{k\geq 1}\,
\bignorm{\Phi\otimes\I_{\Lin(\complex^k)}}_\infty
= \bignorm{\Phi\otimes\I_{\Lin(\Y)}}_\infty.
\end{align*}
By the duality of the trace and spectral norms, it holds that
\begin{equation}
  \label{eq:CB-trace-versus-spectral}
  \triplenorm{\Phi}_1 = \triplenorm{\Phi^{\ast}}_\infty
\end{equation}
for every mapping $\Phi\in\Trans(\X,\Y)$.
In the subsequent sections of the paper, our focus will be on
semidefinite programming formulations of the completely bounded trace
norm $\triplenorm{\cdot}_1$;
interested readers may directly adapt these formulations to ones for
the complete bounded spectral norm by means of the relationship
\eqref{eq:CB-trace-versus-spectral}.
As every operator $X$ having trace norm bounded by 1 can be written as
a convex combination of rank 1 operators taking the form $u v^{\ast}$
for $u$ and $v$ being unit vectors, it follows from the convexity of
norms that
\begin{equation}
  \label{eq:CB-trace-norm-rank-one}
  \triplenorm{\Phi}_1 = 
  \max\bigl\{\bignorm{(\Phi\otimes\I_{\Lin(\X)})(u v^{\ast})}_1\,:\,
    u,v\in\X\otimes\X,\,\norm{u} = \norm{v} = 1\bigr\}.
\end{equation}

Finally, for any two positive semidefinite operators $P,Q\in\Pos(\X)$,
one defines the \emph{fidelity} between $P$ and $Q$ as
\begin{equation} 
  \label{eq:fidelity-definition}
  \fid(P,Q) = \Bignorm{\sqrt{P}\sqrt{Q}}_1.
\end{equation}
For $u,v\in\X\otimes\Y$ being any choice of vectors, it holds that
\begin{equation}
  \label{eq:fidelity-trace-norm}
  \fid\bigl(\tr_{\Y}(u u^{\ast}),\tr_{\Y}(v v^{\ast})\bigr)
  = \norm{\tr_{\X}(v u^{\ast})}_1.
\end{equation}
(It should be noted that the partial traces on the left-hand-side of
the equality in this theorem are taken over the space $\Y$, while
the partial trace on the right-hand-side is taken over $\X$.)
A proof of this identity may be found in \cite{RosgenW05} or
\cite{Watrous08}.

\subsubsection*{Semidefinite programming}

A \emph{semidefinite program}\footnote{
  It should be noted that the above definition differs slightly from
  the one in \cite{Watrous09}, where the equality constraint 
  $\Xi(X) = D$ appears instead as a inequality constraint
  $\Xi(X)\leq D$, and (correspondingly) the dual condition
  $Y\in\Herm(\Y)$ appears as $Y\in\Pos(\Y)$.
  The two forms can easily be converted back and forth, but the one
  above is more convenient for the purposes of this paper.
}
is specified by a triple $(\Xi,C,D)$, where
\begin{mylist}{\parindent}
\item[1.] 
  $\Xi\in\Trans(\X,\Y)$ is a Hermiticity-preserving linear map, and
\item[2.] 
  $C\in\Herm(\X)$ and $D\in\Herm(\Y)$ are Hermitian operators,
\end{mylist}
for $\X = \complex^n$ and $\Y = \complex^m$ denoting spaces as before.
We associate with the triple $(\Xi,C,D)$ two optimization problems,
called the \emph{primal} and \emph{dual} problems, as follows:
\begin{center}
  \begin{minipage}[t]{2.6in}
    \centerline{\underline{Primal problem}}\vspace{-7mm}
    \begin{align*}
      \text{maximize:}\quad & \ip{C}{X}\\
      \text{subject to:}\quad & \Xi(X) = D,\\
      & X\in\Pos(\X).
    \end{align*}
  \end{minipage}
  \hspace*{13mm}
  \begin{minipage}[t]{2.6in}
    \centerline{\underline{Dual problem}}\vspace{-7mm}
    \begin{align*}
      \text{minimize:}\quad & \ip{D}{Y}\\
      \text{subject to:}\quad & \Xi^{\ast}(Y) \geq C,\\
      & Y\in\Herm(\Y).
    \end{align*}
  \end{minipage}
\end{center}

An operator $X\in\Pos(\X)$ satisfying $\Xi(X) = D$ is said to be
\emph{primal feasible}, and an operator $Y\in\Herm(\Y)$ satisfying
$\Xi^{\ast}(Y) \geq C$ is said to be \emph{dual feasible}.
We let $\P$ and $\D$ denote the sets of primal and dual feasible
operators , respectively:
\[
\P = \bigl\{X\in\Pos(\X)\,:\,\Xi(X)=D\bigr\} 
\qquad\text{and}\qquad
\D = \bigl\{Y\in\Herm(\Y)\,:\,\Xi^{\ast}(Y) \geq C\bigr\}.
\]
The linear functions $X \mapsto \ip{C}{X}$ and $Y \mapsto \ip{D}{Y}$
are referred to as the primal and dual \emph{objective functions},
which take real number values all choices of $X\in\P$ and $Y\in\D$
(or, more generally, over all choices of $X\in\Herm(\X)$ and
$Y\in\Herm(\Y)$).
The \emph{primal optimum} and \emph{dual optimum}
are defined as
\[
\alpha = \sup_{X\in\P}\ip{C}{X}
\qquad\text{and}\qquad
\beta = \inf_{Y\in\D}\ip{D}{Y},
\]
respectively.
The values $\alpha$ and $\beta$ may be finite or infinite, and by
convention we define $\alpha = -\infty$ if $\P=\varnothing$ and 
$\beta = \infty$ if $\D=\varnothing$.
If an operator $X\in\P$ satisfies $\ip{C}{X} = \alpha$ we say that $X$
is an \emph{optimal primal solution}, or that $X$ achieves the
primal optimum.
Likewise, if $Y\in\D$ satisfies $\ip{D}{Y} = \beta$ we say that $Y$ is
an \emph{optimal dual solution}, or that $Y$ achieves the dual optimal.

For every semidefinite program it holds that $\alpha \leq \beta$,
which is a fact known as \emph{weak duality}.
The condition $\alpha = \beta$, known as \emph{strong duality}, may
fail to hold for some semidefinite programs---but, for a wide range of
semidefinite programs that arise in practice, strong duality does hold.
The following theorem provides a condition (in both a primal and dual
form) that implies strong duality.

\begin{theorem}[Slater's theorem for semidefinite programs]
  \label{theorem:Slater}
  The following implications hold for every semidefinite program
  $(\Xi,C,D)$.
  \begin{mylist}{\parindent}
  \item[1.]
    If $\P\not=\varnothing$ and there exists a Hermitian operator $Y$
    for which $\Xi^{\ast}(Y) > C$, then $\alpha = \beta$ and there
    exists a primal feasible operator $X\in\P$ for which
    $\ip{C}{X}=\alpha$.
  \item[2.]
    If $\D\not=\varnothing$ and there exists a positive definite
    operator $X>0$ for which $\Xi(X) = D$, then $\alpha = \beta$ and
    there exists a dual feasible operator $Y\in\D$ for which
    $\ip{D}{Y}=\beta$.
  \end{mylist}
\end{theorem}
\noindent
The condition that some operator $X>0$ satisfies $\Xi(X) = D$ is
called \emph{strict primal feasibility}, while the condition that some
operator $Y\in\Herm(\Y)$ satisfies $\Xi^{\ast}(Y) > C$ is called
\emph{strict dual feasibility};
in both cases, the ``strictness'' concerns the positive semidefinite
ordering.

\section{A semidefinite program for the maximum output fidelity}
\label{sec:max-output-fidelity}

The first semidefinite programming formulation of the completely
bounded trace norm to be presented is based on a characterization of
the completely bounded trace norm in terms of the fidelity function,
together with a simple semidefinite program for the fidelity function
itself.

\subsection{A semidefinite program for the fidelity function}

We will begin by presenting a semidefinite programming
characterization of the fidelity $\fid(P,Q)$ between two positive
semidefinite operators $P,Q\in\Pos(\X)$, for $\X = \complex^n$.
The same semidefinite programming characterization of the fidelity was
independently discovered by Nathan Killoran \cite{Killoran12}.

The semidefinite program is given by the triple $(\Xi,C,D)$, where
$\Xi:\Lin(\X\oplus\X)\rightarrow\Lin(\X\oplus\X)$ is defined as
\[
\Xi\!
\begin{pmatrix} 
  X_{1,1} & X_{1,2}\\ 
  X_{2,1} & X_{2,2}
\end{pmatrix}
=
\begin{pmatrix} 
  X_{1,1} & 0\\ 
  0 & X_{2,2}
\end{pmatrix}
\]
for all $X_{1,1},X_{1,2},X_{2,1},X_{2,2}\in\Lin(\X)$, and
$C,D\in\Herm(\X\oplus\X)$ are defined as
\[
C = \frac{1}{2}\begin{pmatrix}
  0 & \I\\
  \I & 0
\end{pmatrix}
\quad\quad\text{and}\quad\quad
D = \begin{pmatrix}
  P & 0\\
  0 & Q
\end{pmatrix}.
\]
The primal and dual problems associated with this semidefinite program
may, after some simplifications, be expressed as follows:
\begin{center}
  \begin{minipage}[t]{2.6in}
    \centerline{\underline{Primal problem}}\vspace{-5mm}
    \begin{align*}
      \text{maximize:}\quad & \frac{1}{2}\tr(X) + \frac{1}{2}\tr(X^{\ast}) \\
      \text{subject to:}\quad & 
      \begin{pmatrix}
        P & X\\ X^{\ast} & Q
      \end{pmatrix}\geq 0\\
      & X\in\Lin(\X).
    \end{align*}
  \end{minipage}
  \hspace*{13mm}
  \begin{minipage}[t]{2.6in}
    \centerline{\underline{Dual problem}}\vspace{-5mm}
    \begin{align*}
      \text{minimize:}\quad & \frac{1}{2}\ip{P}{Y}+ \frac{1}{2}\ip{Q}{Z}\\
      \text{subject to:}\quad &
      \begin{pmatrix}
        Y & -\I\\
        -\I & Z
      \end{pmatrix}
      \geq 0\\
      & Y, Z \in\Herm(\X).
    \end{align*}
  \end{minipage}
\end{center}

\subsubsection*{Strong duality}

Strong duality for the semidefinite program $(\Xi,C,D)$ may be
verified through an application of Slater's theorem, using the fact
that the primal problem is feasible and the dual problem is strictly
feasible.
In particular, the operator
\[
\begin{pmatrix}
  P & 0\\
  0 & Q
\end{pmatrix}
\]
is primal feasible, which implies that $\P\not=\varnothing$.
For the dual problem, the operator
\[
\begin{pmatrix}
  \I & 0\\
  0 & \I
\end{pmatrix}
\]
is strictly feasible, as
\[
\Xi^{\ast}\!
\begin{pmatrix}
  \I & 0\\
  0 & \I
\end{pmatrix}
=
\begin{pmatrix}
  \I & 0\\
  0 & \I
\end{pmatrix}
>
\frac{1}{2}
\begin{pmatrix}
  0 & \I\\
  \I & 0
\end{pmatrix}.
\]
By Slater's theorem, we have strong duality, and moreover the primal
optimum is achieved by some choice of a primal feasible operator.

It so happens that strict primal feasibility may fail to hold: if
either of $P$ or $Q$ is not positive definite, it cannot hold that
\[
\begin{pmatrix}
  P & X\\ X^{\ast} & Q
\end{pmatrix}> 0.
\]
One cannot conclude from this fact that the optimal dual value will
not be achieved---but indeed this is the case for some choices of $P$
and $Q$.
If $P$ and $Q$ are positive definite, however, then strict primal
feasibility does hold, and the existence of an optimal dual solution
follows from Slater's theorem.

\subsubsection*{Optimal value}

One may prove that the optimal value of the semidefinite program
described above is equal to $\fid(P,Q)$ by making use of the following
fact (stated as Theorem IX.5.9 in \cite{Bhatia97}).

\begin{lemma}
  \label{lemma:block-positive}
  Let $P,Q\in\Pos(\complex^n)$ be positive semidefinite operators and
  let $X\in\Lin(\complex^n)$ be any operator.
  It holds that
  \begin{equation}
    \label{eq:block-positive-lemma}
    \begin{pmatrix}
      P & X\\
      X^{\ast} & Q
    \end{pmatrix}
    \in \Pos(\complex^n\oplus\complex^n)
  \end{equation}
  if and only if $X = \sqrt{P}K\sqrt{Q}$ for $K\in\Lin(\complex^n)$
  satisfying $\norm{K}_{\infty} \leq 1$.
\end{lemma}

\noindent
It follows from this lemma that for feasible solutions to the primal
problem, the variable $X\in\Lin(\X)$ (in the simplified form of the
primal problem) is free to range precisely over those operators given
by $\sqrt{P}K\sqrt{Q}$ for $K\in\Lin(\X)$ satisfying
$\norm{K}_{\infty} \leq 1$.
The primal optimum is therefore given by
\begin{multline*}
  \qquad
  \sup_{K}
  \biggl(\frac{1}{2}\tr\Bigl(\sqrt{P} K \sqrt{Q}\Bigr)
  + \frac{1}{2}\tr\Bigl(\sqrt{Q} K^{\ast} \sqrt{P}\Bigr)
  \biggr)
  =
  \sup_{K}
  \Re\Bigl(\tr\Bigl(\sqrt{Q} K^{\ast} \sqrt{P}\Bigr)\Bigr)\\
  =
  \sup_{K}\,
  \Bigabs{\Bigl(\tr\Bigl(\sqrt{Q} K^{\ast} \sqrt{P}\Bigr)\Bigr)}
  =
  \sup_{K}\,
  \Bigabs{\Bigip{K}{\sqrt{P}\sqrt{Q}}}
  =
  \Bignorm{\sqrt{P}\sqrt{Q}}_1
  = \fid(P,Q),
  \qquad
\end{multline*}
where each supremum is over the set
$\{K\in\Lin(\X)\,:\,\norm{K}_{\infty}\leq 1\}$.

By strong duality, the dual optimum is also equal to $\fid(P,Q)$.
An alternate way to prove this fact begins with the observation that
the dual optimum is equal to
\begin{equation} \label{eq:fidelity-dual-simplified}
  \inf_{Y\in\Pd(\X)}
  \biggl(\frac{1}{2}\bigip{P}{Y}+ \frac{1}{2}\bigip{Q}{Y^{-1}}\biggr)
\end{equation}
This expression follows from the observation that, for every
$Y,Z\in\Herm(\X)$, it holds that
\[
\begin{pmatrix} 
  Y & -\I\\
  -\I & Z 
\end{pmatrix}
\in\Pos(\X\otimes\X)
\]
if and only if $Y,Z\in\Pd(\X)$ and $Z \geq Y^{-1}$, together with the
assumption that $Q$ is positive semidefinite.
Now, the fact that the dual optimum is equal to $\fid(P,Q)$ follows
from a theorem known as \emph{Alberti's theorem}.

\begin{theorem}[Alberti] \label{theorem:alberti}
  Let $\X = \complex^n$ and let $P,Q\in\Pos(\X)$ be positive
  semidefinite operators.
  It holds that
  \[
  \left(\fid(P,Q)\right)^2 = \inf_{Y\in\Pd(\X)} \bigip{P}{Y} \bigip{Q}{Y^{-1}}.
  \]
\end{theorem}

\noindent
To see that Alberti's theorem implies that the expression
\eqref{eq:fidelity-dual-simplified} is equal to $\fid(P,Q)$, note
first that the arithmetic-geometric mean inequality implies
that
\[
\frac{1}{2}\bigip{P}{Y} + \frac{1}{2}\bigip{Q}{Y^{-1}} 
\geq \sqrt{\ip{P}{Y}\ip{Q}{Y^{-1}}}
\]
for every $Y\in\Pd(\X)$, with equality if and only if 
$\ip{P}{Y} = \ip{Q}{Y^{-1}}$.
It follows that
\[
\inf_{Y\in\Pd(\X)}
\biggl(\frac{1}{2}\bigip{P}{Y}+ \frac{1}{2}\bigip{Q}{Y^{-1}}\biggr)
\geq \fid(P,Q).
\]
Moreover, for an arbitrary choice of $Y\in\Pd(\X)$, one may choose
$\lambda > 0$ so that
\[
\bigip{P}{\lambda Y} = \bigip{Q}{(\lambda Y)^{-1}}
\]
and therefore
\[
\frac{1}{2}\bigip{P}{\lambda Y} + \frac{1}{2}\bigip{Q}{(\lambda Y)^{-1}} 
= \sqrt{\ip{P}{\lambda Y}\ip{Q}{(\lambda Y)^{-1}}} = 
\sqrt{\ip{P}{Y}\ip{Q}{Y^{-1}}}.
\]
Thus, 
\[
\inf_{Y\in\Pd(\X)}
\biggl(\frac{1}{2}\bigip{P}{Y}+ \frac{1}{2}\bigip{Q}{Y^{-1}}\biggr)
= \fid(P,Q).
\]

By reversing this argument, an alternate proof of Alberti's theorem
based on semidefinite programming duality is obtained.
A similar observation was made in \cite{Watrous09} based on a
different semidefinite programming formulation of the fidelity.

\subsection{Maximum output fidelity characterization of the completely
  bounded trace norm}

Next, we recall a known characterization of the completely bounded
trace norm in terms of the fidelity function, which makes use of the
following definition.

\begin{definition}
  Let $\X = \complex^n$ and $\Z= \complex^k$, and let
  $\Psi_0,\Psi_1\in\Trans(\X,\Z)$ be positive maps.
  The \emph{maximum output fidelity} between $\Psi_0$ and
  $\Psi_1$ is defined as 
  \[
  \fid_{\textup{max}}(\Psi_0,\Psi_1)
  =
  \max\bigl\{\fid(\Psi_0(\rho_0),\Psi_1(\rho_1))\,:\,
  \rho_0,\rho_1\in\Density(\X)\bigr\}.
  \]
\end{definition}

\noindent
The characterization (which appears as an exercise in
\cite{KitaevSV02} and is a corollary of a slightly more general result
proved in \cite{Watrous08}) is given by the following theorem.

\begin{theorem}
  \label{theorem:CB-trace-fidelity}
  Let $\X = \complex^n$, $\Y= \complex^m$, and $\Z = \complex^k$, let
  $A_0,A_1\in\Lin(\X,\Y\otimes\Z)$ be operators, and let
  $\Psi_0,\Psi_1\in\Trans(\X,\Z)$ and $\Phi\in\Trans(\X,\Y)$ be
  mappings defined by the equations
  \[
  \Psi_0(X) = \tr_{\Y}\bigl(A_0 X A_0^{\ast}\bigr),
  \qquad
  \Psi_1(X) = \tr_{\Y}\bigl( A_1 X A_1^{\ast}\bigr),
  \qquad\text{and}\qquad
  \Phi(X) = \tr_{\Z} \bigl(A_0 X A_1^{\ast}\bigr),
  \]
  for all $X\in\Lin(\X)$.
  It holds that $\triplenorm{\Phi}_1=\fid_{\mathrm{max}}(\Psi_0,\Psi_1)$.
\end{theorem}

\begin{proof}
  For $\W = \complex^n$ and any choice of vectors
  $u_0,u_1\in\X\otimes\W$, one has
  \begin{align*}
    \tr_{\Y\otimes\W} 
    \bigl( (A_0\otimes \I_{\W}) 
    u_0 u_0^{\ast} (A_0\otimes \I_{\W})^{\ast}\bigr)
    & = \Psi_0\bigl(\tr_{\W}\bigl(u_0 u_0^{\ast}\bigr)\bigr),\\
    \tr_{\Y\otimes\W} 
    \bigl( (A_1\otimes \I_{\W}) 
    u_1 u_1^{\ast} (A_1\otimes \I_{\W})^{\ast}\bigr)
    & = \Psi_1\bigl(\tr_{\W}\bigl(u_1 u_1^{\ast}\bigr)\bigr),
  \end{align*}
  and therefore, by \eqref{eq:fidelity-trace-norm}, it holds that
  \[
  \bignorm{\tr_{\Z}\bigl((A_0\otimes \I_{\W})u_0 
    u_1^{\ast}(A_1\otimes\I_{\W})^{\ast}\bigr)}_1
  = \fid\bigl(\Psi_0(\tr_{\W}(u_0 u_0^{\ast})),
  \Psi_1(\tr_{\W}(u_1 u_1^{\ast}))\bigr).
  \]
  Consequently
  \begin{align*}
    \triplenorm{\Phi}_1 
    & =
    \op{max}
    \bigl\{
    \bignorm{\tr_{\Z}\bigl((A_0\otimes \I_{\W})u_0 
      u_1^{\ast}(A_1\otimes\I_{\W})^{\ast}\bigr)}_1
    \,:\,
    u_0, u_1 \in \X\otimes\W,\:
    \norm{u_0} = \norm{u_1} = 1
    \bigr\}\\
    & =
    \op{max}
    \bigl\{
    \fid\bigl(
    \Psi_0(\tr_{\W}(u_0 u_0^{\ast})),
    \Psi_1(\tr_{\W}(u_1 u_1^{\ast}))
    \bigl)\,:\,
    u_0, u_1 \in \X\otimes\W,\:
    \norm{u_0} = \norm{u_1} = 1
    \bigr\}\\
    & =
    \op{max}
    \bigl\{
    \fid(\Psi_0(\rho_0),\Psi_1(\rho_1))\,:\,
    \rho_0,\rho_1\in\Density(\X)\bigr\}\\
    & = \fid_{\mathrm{max}}(\Psi_0,\Psi_1)
  \end{align*}
  as required.
\end{proof}

\subsection{A semidefinite program for the maximum output fidelity}

Theorem~\ref{theorem:CB-trace-fidelity}, when combined with the
semidefinite program for the fidelity discussed at the beginning of
the present section, leads to a semidefinite program for the
completely bounded trace norm, as is now described.

Let $\X=\complex^n$ and $\Y=\complex^m$, and suppose that a mapping
$\Phi\in\Trans(\X,\Y)$ is given as
\begin{equation}
  \Phi(X) = \tr_{\Z} \bigl(A_0 X A_1^{\ast}\bigr)
\end{equation}
for all $X\in\Lin(\X)$, where $\Z = \complex^k$ and
$A_0,A_1\in\Lin(\X,\Y\otimes\Z)$ are operators.
An expression of this form is sometimes known as a
\emph{Stinespring representation} of $\Phi$, and such a representation
always exists (provided that $k$ is sufficiently large; $k$ must be
at least $mn$ in the worst case).

Now, define completely positive mappings
$\Psi_0,\Psi_1\in\Trans(\X,\Z)$ as
\[
\Psi_0(X) = \tr_{\Y}\bigl(A_0 X A_0^{\ast}\bigr)
\qquad\text{and}\qquad
\Psi_1(X) = \tr_{\Y}\bigl( A_1 X A_1^{\ast}\bigr)
\]
for all $X\in\Lin(\X)$.
The semidefinite program to be considered is specified by the triple
$(\Xi,C,D)$, where $\Xi : \Lin(\X \oplus \X \oplus \Z \oplus \Z) \rightarrow
\Lin(\complex \oplus \complex \oplus \Z \oplus \Z)$ is a
Hermiticity-preserving mapping defined as
\begin{equation}
  \Xi\begin{pmatrix}
  X_0 & \cdot & \cdot & \cdot\\
  \cdot & X_1 & \cdot & \cdot\\
  \cdot & \cdot & Z_0 & \cdot\\
  \cdot & \cdot & \cdot & Z_1
  \end{pmatrix}
  =
  \begin{pmatrix}
    \tr(X_0) & 0 & 0 & 0\\
    0 & \tr(X_1) & 0 & 0\\
    0 & 0 & Z_0 - \Psi_0(X_0) & 0\\
    0 & 0 & 0 & Z_1 - \Psi_1(X_1)
  \end{pmatrix},
\end{equation}
where dots represent operators on appropriately chosen spaces upon
which $\Xi$ does not depend, and 
$C\in\Herm(\X \oplus \X \oplus \Z \oplus \Z)$ and
$D\in\Herm(\complex\oplus\complex\oplus\Z\oplus\Z)$ are defined as
\begin{equation}
  C = \frac{1}{2}\begin{pmatrix}
    0 & 0 & 0 & 0\\
    0 & 0 & 0 & 0\\
    0 & 0 & 0 & \I\\
    0 & 0 & \I & 0
  \end{pmatrix}
\qquad\text{and}\qquad
  D = \begin{pmatrix}
    1 & 0 & 0 & 0\\
    0 & 1 & 0 & 0\\
    0 & 0 & 0 & 0\\
    0 & 0 & 0 & 0
  \end{pmatrix}.
\end{equation}
The adjoint of the mapping $\Xi$ is given by
\[
\Xi^{\ast}
\begin{pmatrix}
  \lambda_0 & \cdot & \cdot & \cdot\\
  \cdot & \lambda_1 & \cdot & \cdot\\
  \cdot & \cdot & Y_0 & \cdot\\
  \cdot & \cdot & \cdot & Y_1
\end{pmatrix}
=
\begin{pmatrix}
  \lambda_0 \I_{\X} - \Psi_0^{\ast}(Y_0) & 0 & 0 & 0\\
  0 & \lambda_1 \I_{\X} - \Psi_1^{\ast}(Y_1) & 0 & 0\\
  0 & 0 & Y_0 & 0\\
  0 & 0 & 0 & Y_1
\end{pmatrix}.
\]
After a simplification of the primal and dual problems associated with
$(\Xi,C,D)$, one obtains equivalent primal and dual problems as follows:
\begin{center}
  \begin{minipage}[t]{3in}
    \centerline{\underline{Primal problem}}\vspace{-5mm}
    \begin{align*}
    \text{maximize:}\quad &
    \frac{1}{2}\tr(X) + \frac{1}{2}\tr(X^{\ast})\\[1mm]
    \text{subject to:}\quad & 
    \begin{pmatrix}
      \Psi_0(\rho_0) & X\\[2mm]
      X^{\ast} & \Psi_1(\rho_1)
    \end{pmatrix}
    \geq 0\\
    & \rho_0,\rho_1\in\Density(\X)\\
    & X\in\Lin(\Z).
    \end{align*}
  \end{minipage}
  \begin{minipage}[t]{3in}
    \centerline{\underline{Dual problem}}\vspace{-5mm}
    \begin{align*}
    \text{minimize:}\quad &
    \frac{1}{2}\bignorm{\Psi_0^{\ast}(Y)}_{\infty} +
    \frac{1}{2}\bignorm{\Psi_1^{\ast}(Y^{-1})}_{\infty}\\[1mm]
    \text{subject to:}\quad &
    Y \in \Pd(\Z).
    \end{align*}
  \end{minipage}
\end{center}

\subsubsection*{Strong duality}

To prove that strong duality holds for the semidefinite program above,
it suffices to prove that the primal problem is feasible and the dual
problem is strictly feasible.
Primal feasibility is easily checked:
one may verify that the operator
\[
\begin{pmatrix}
  \rho_0 & 0 & 0 & 0\\
  0 & \rho_1 & 0 & 0\\
  0 & 0 & \Psi_0(\rho_0) & 0\\
  0 & 0 & 0 & \Psi_1(\rho_1)
\end{pmatrix}
\]
is primal feasible for any choice of density operators
$\rho_0,\rho_1\in\Density(\X)$.
To verify that strict dual feasibility holds, one may consider the
operator
\[
\begin{pmatrix}
  \lambda_0 & 0 & 0 & 0\\
  0 & \lambda_1 & 0 & 0\\
  0 & 0 & \I_{\Z} & 0\\
  0 & 0 & 0 & \I_{\Z}
\end{pmatrix}
\]
for any choice of real numbers 
$\lambda_0 > \norm{\Psi_0^{\ast}(\I_{\Z})}_{\infty}$ and
$\lambda_1 > \norm{\Psi_1^{\ast}(\I_{\Z})}_{\infty}$.
By Slater's theorem, strong duality follows.

\subsubsection*{Optimal value}

For any fixed choice of $\rho_0,\rho_1\in\Density(\X)$, one has that
the maximum value of the primal objective function
\[
\frac{1}{2}\tr(X) + \frac{1}{2}\tr(X^{\ast})
\]
subject to the constraint
\[
\begin{pmatrix}
  \Psi_0(\rho_0) & X\\[2mm]
  X^{\ast} & \Psi_1(\rho_1)
\end{pmatrix}
\geq 0
\]
is equal to $\fid(\Psi_0(\rho_0),\Psi_1(\rho_1))$, by the same analysis
that was used to determine the primal optimum for the semidefinite
program for the fidelity function.
Maximizing over all choices of density operators
$\rho_0,\rho_1\in\Density(\X)$ gives
$\fid_{\textup{max}}(\Psi_0,\Psi_1)$, which equals
$\triplenorm{\Phi}_1$ by Theorem~\ref{theorem:CB-trace-fidelity}.

\section{A semidefinite program for the completely bounded trace norm
  from a mapping's Choi-Jamio{\l}kowski representation}
\label{sec:Choi}

In this section an alternate semidefinite program for the completely
bounded trace norm is presented.
Whereas the semidefinite program from the previous section is obtained
from a Stinespring representation of a given mapping, the semidefinite
program in this section is obtained from the Choi-Jamio{\l}kowski
representation of a given mapping.

While the two semidefinite programming formulations are different,
they are closely related.
As for the semidefinite programs for the fidelity and the completely
bounded trace norm in the previous section,
Lemma~\ref{lemma:block-positive} provides a key tool through which the
semidefinite program given in this section may be analyzed.

\subsection{Choi-Jamio{\l}kowski representations and the completely
  bounded trace norm}

Let $\X = \complex^n$ and $\Y = \complex^m$, and assume that
$\Phi\in\Trans(\X,\Y)$ is a given mapping.
The \emph{Choi-Jamio{\l}kowski representation} of $\Phi$ is the
operator $J(\Phi) \in \Lin(\Y\otimes\X)$ defined as
\[
J(\Phi) = \sum_{1\leq i,j \leq n} \Phi(E_{i,j}) \otimes E_{i,j}.
\]
An equivalent expression is
\[
J(\Phi) = (\Phi\otimes\I_{\Lin(\X)})(\op{vec}(\I_{\X})
\op{vec}(\I_{\X})^{\ast}),
\]
where the vec-mapping is the linear mapping defined by the action
\[
\vec(E_{i,j}) = e_i \otimes e_j,
\]
extended by linearity to arbitrary operators.

One identity connecting the vec-mapping to the Choi-Jamio{\l}kowski
representation of a mapping is the following one, which holds for all
choices of $A,B\in\Lin(\X)$:
\begin{equation}
  \label{eq:Choi-evaluation}
  \bigl( \I_{\Y} \otimes A^{\t}\bigr)
  J(\Phi) \bigl( \I_{\Y} \otimes \overline{B}\bigr)
  =
  \bigl(\Phi \otimes \I_{\Lin(\X)}\bigr)
  \bigl(\op{vec}(A) \op{vec}(B)^{\ast}\bigr).
\end{equation}
Through this identity, an alternate expression for the completely
bounded trace norm is obtained, as stated by the following theorem.

\begin{theorem}
  \label{theorem:CB-trace-norm-Choi}
  Let $\X = \complex^n$ and $\Y = \complex^m$, and let
  $\Phi\in\Trans(\X,\Y)$ be a linear mapping.
  It holds that
  \[
  \triplenorm{\Phi}_1 = 
  \max\Bigl\{
  \bignorm{
    \bigl(\,\I_{\Y}\otimes\sqrt{\rho_0}\,\bigr)J(\Phi)
    \bigl(\,\I_{\Y}\otimes\sqrt{\rho_1}\,\bigr)}_1\,:\,
  \rho_0,\rho_1\in\Density(\X)\Bigr\}.
\]
\end{theorem}

\begin{proof}
  By \eqref{eq:CB-trace-norm-rank-one} together with
  \eqref{eq:Choi-evaluation} it holds that
  \[
  \triplenorm{\Phi}_1
  = \max
  \bigl\{
  \bignorm{
    \bigl( \I_{\Y} \otimes A^{\t}\bigr)
    J(\Phi) \bigl( \I_{\Y} \otimes \overline{B}\bigr)}_1\,:\,
  A,B\in\Lin(\X),\,\norm{A}_2 = \norm{B}_2 = 1\bigr\}.
  \]
  By the polar decomposition, every operator $X\in\Lin(\X)$ with
  $\norm{X}_2 = 1$ may be written as $X = \sqrt{\sigma} U$ for some
  choice of $\sigma\in\Density(\X)$ and $U\in\Unitary(\X)$.
  By the unitary invariance of the trace norm, the theorem follows.
\end{proof}

\subsection{A semidefinite program from
  Theorem~\ref{theorem:CB-trace-norm-Choi}}

The semidefinite program to be considered is specified by the triple
$(\Xi,C,D)$, where
\[
\Xi\in\Trans(\X \oplus \X \oplus (\Y\otimes\X) \oplus (\Y\otimes\X),
\complex \oplus \complex \oplus (\Y\otimes\X) \oplus (\Y\otimes\X))
\]
is a Hermiticity-preserving mapping defined as
\begin{equation}
  \Xi\begin{pmatrix}
  X_0 & \cdot & \cdot & \cdot\\
  \cdot & X_1 & \cdot & \cdot\\
  \cdot & \cdot & Z_0 & \cdot\\
  \cdot & \cdot & \cdot & Z_1
  \end{pmatrix}
  =
  \begin{pmatrix}
    \tr(X_0) & 0 & 0 & 0\\
    0 & \tr(X_1) & 0 & 0\\
    0 & 0 & Z_0 - \I_{\Y} \otimes X_0 & 0\\
    0 & 0 & 0 & Z_1 - \I_{\Y}\otimes X_1
  \end{pmatrix}
\end{equation}
and $C\in\Herm(\X\oplus\X\oplus(\Y\otimes\X)\oplus(\Y\otimes\X))$ and
$D\in\Herm(\complex\oplus\complex\oplus(\Y\otimes\X)\oplus(\Y\otimes\X))$
are defined as 
\begin{equation}
  C = \frac{1}{2}\begin{pmatrix}
    0 & 0 & 0 & 0\\
    0 & 0 & 0 & 0\\
    0 & 0 & 0 & J(\Phi)\\
    0 & 0 & J(\Phi)^{\ast} & 0
  \end{pmatrix}
\qquad\text{and}\qquad
  D = \begin{pmatrix}
    1 & 0 & 0 & 0\\
    0 & 1 & 0 & 0\\
    0 & 0 & 0 & 0\\
    0 & 0 & 0 & 0
  \end{pmatrix}.
\end{equation}
The adjoint of the mapping $\Xi$ is given by
\[
\Xi^{\ast}
\begin{pmatrix}
  \lambda_0 & \cdot & \cdot & \cdot\\
  \cdot & \lambda_1 & \cdot & \cdot\\
  \cdot & \cdot & Y_0 & \cdot\\
  \cdot & \cdot & \cdot & Y_1
\end{pmatrix}
=
\begin{pmatrix}
  \lambda_0 \I_{\X} - \tr_{\Y}(Y_0) & 0 & 0 & 0\\
  0 & \lambda_1 \I_{\X} - \tr_{\Y}(Y_1) & 0 & 0\\
  0 & 0 & Y_0 & 0\\
  0 & 0 & 0 & Y_1
\end{pmatrix}.
\]
After a simplification of the primal and dual problems associated with
$(\Xi,C,D)$, one obtains equivalent primal and dual problems as follows:
\begin{center}
  \begin{minipage}[t]{3in}
    \centerline{\underline{Primal problem}}\vspace{-4mm}
    \begin{align*}
    \text{maximize:}\quad &
    \frac{1}{2}\ip{J(\Phi)}{X} + \frac{1}{2}\ip{J(\Phi)^{\ast}}{X^{\ast}}
    \\[2mm]
    \text{subject to:}\quad & 
    \begin{pmatrix}
      \I_{\Y} \otimes \rho_0 & X\\[2mm]
      X^{\ast} & \I_{\Y}\otimes \rho_1
    \end{pmatrix}
    \geq 0\\
    & \rho_0,\rho_1\in\Density(\X)\\
    & X\in\Lin(\Y\otimes\X)
    \end{align*}
  \end{minipage}
  \qquad
  \begin{minipage}[t]{3in}
    \centerline{\underline{Dual problem}}\vspace{-4mm}
    \begin{align*}
    \text{minimize:}\quad & 
    \frac{1}{2} \bignorm{\tr_{\Y}(Y_0)}_{\infty}
    + \frac{1}{2} \bignorm{\tr_{\Y}(Y_1)}_{\infty}\\[2mm]
    \text{subject to:}\quad &
    \begin{pmatrix}
      Y_0 & -J(\Phi)\\[2mm]
      -J(\Phi)^{\ast} & Y_1
    \end{pmatrix}
    \geq 0\\
    & Y_0, Y_1 \in \Pos(\Y\otimes\X)
    \end{align*}
  \end{minipage}
\end{center}

\subsubsection*{Strong duality}

Similar to the semidefinite programs discussed in the previous section,
strong duality is easily established for the semidefinite program
described above by the use of Slater's theorem.
In fact, strict primal and strict dual feasibility hold for all
choices of $\Phi$; so that, in addition to strong duality, the primal
and dual optima are achieved by feasible solutions in both cases.
An example of a strictly feasible primal solution is
\[
\begin{pmatrix}
  X & 0 & 0 & 0\\
  0 & X & 0 & 0\\
  0 & 0 & Z & 0\\
  0 & 0 & 0 & Z
  \end{pmatrix}
\]
for
\[
X = \frac{\I_{\X}}{\dim(\X)}
\qquad\text{and}\qquad
Z = \frac{\I_{\Y}\otimes\I_{\X}}{\dim(\X)},
\]
while an example of a strictly feasible dual solution is
\[
\begin{pmatrix}
  \lambda & 0 & 0 & 0\\
  0 & \lambda & 0 & 0\\
  0 & 0 & Y & 0\\
  0 & 0 & 0 & Y
\end{pmatrix}
\]
for
\[
Y = \biggl(\frac{\norm{J(\Phi)}_{\infty}}{2} + 1\biggr)\, \I_{\Y\otimes\X}
\qquad\text{and}\qquad
\lambda 
= 1 + \biggl(\frac{\norm{J(\Phi)}_{\infty}}{2} + 1\biggr)\, \dim(\Y).
\]

\subsubsection*{Optimal value}

For any choice of density operators
$\rho_0,\rho_1\in\Density(\X)$, it holds that
\begin{equation}
  \begin{pmatrix}
    \I_{\Y} \otimes \rho_0 & X\\[2mm]
    X^{\ast} & \I_{\Y}\otimes \rho_1
  \end{pmatrix}
  \geq 0
\end{equation}
if and only if
\begin{equation}
  X = \bigl(\,\I_{\Y}\otimes\sqrt{\rho_0}\,\bigr) K
  \bigl(\,\I_{\Y}\otimes\sqrt{\rho_1}\,\bigr)
\end{equation}
for some choice of an operator $K\in\Lin(\Y\otimes\X)$ satisfying
$\norm{K}_{\infty} \leq 1$, as follows from Lemma~\ref{lemma:block-positive}.
The primal optimum is therefore given by
\begin{multline*}
  \quad
  \sup_{K,\rho_0,\rho_1}
  \Re\Bigl(
  \bigip{J(\Phi)}{\bigl(\,\I_{\Y}\otimes\sqrt{\rho_0}\,\bigr) K
    \bigl(\,\I_{\Y}\otimes\sqrt{\rho_1}\,\bigr)}\Bigr)
  \biggr)
  = 
  \sup_{\rho_0,\rho_1}
  \Bignorm{
    \bigl(\,\I_{\Y}\otimes\sqrt{\rho_1}\,\bigr)J(\Phi)^{\ast}
    \bigl(\,\I_{\Y}\otimes\sqrt{\rho_0}\,\bigr)}_1\\
  = 
  \sup_{\rho_0,\rho_1}
  \Bignorm{
    \bigl(\,\I_{\Y}\otimes\sqrt{\rho_0}\,\bigr)J(\Phi)
    \bigl(\,\I_{\Y}\otimes\sqrt{\rho_1}\,\bigr)}_1
  = \triplenorm{\Phi}_1,
  \quad
\end{multline*}
where supremums are taken over all $K\in\Lin(\X)$ with
$\norm{K}_{\infty} \leq 1$ and $\rho_0,\rho_1\in\Density(\X)$, and
where the last equality follows from
Theorem~\ref{theorem:CB-trace-norm-Choi}.

\section{Remarks on the complexity of approximating optimal solutions
  to the semidefinite programs}

Suppose that $(\Xi,C,D)$ is an instance of one of the semidefinite
programs described above, either for the maximum output fidelity
characterization or the Choi-Jamio{\l}kowski representation
characterization of the completely bounded trace norm.
It is natural to ask whether an approximation to the optimal value of
this semidefinite program can be efficiently computed (under the
assumption, let us say, that the complex numbers specifying $\Xi$,
$C$, and $D$ have rational real and imaginary parts whose numerators
and denominators are represented as integers in binary notation).

From a practical viewpoint, algorithms employing
\emph{interior point methods} represent a sensible approach for
computing the optimum value of these semidefinite programs
\cite{Alizadeh95,deKlerk02}.
The CVX software package \cite{GrantB09} for the MATLAB numerical
computing environment allows one to solve these semidefinite programs
efficiently with minimal coding requirements.

For the sake of obtaining rigorous statements about the
polynomial-time solvability of the semidefinite programs
(and perhaps not much more than that), the \emph{ellipsoid method} is
a more attractive alternative, applied specifically to the dual
formulations of the semidefinite programs.
When considering the applicability of the ellipsoid method, it is
helpful to consider the following set, for $\D\subseteq\Herm(\Y)$
denoting the dual feasible set of $(\Xi,C,D)$ and $\varepsilon>0$
being a positive real number:
\[
\D^{\circ}_{\varepsilon} = \bigl\{Y\in\Herm(\Y)\,:\,
Y + H \in \D\;\text{for all $H\in\Herm(\Y)$ satisfying 
  $\norm{H}_2\leq\varepsilon$}\bigr\}.
\]
Intuitively speaking, $\D_{\varepsilon}^{\circ}$ contains every
operator in the interior of the dual feasible set that is not too
close to the boundary of that set.

It has already been demonstrated that $\D^{\circ}_{\varepsilon}$ is
nonempty for some choice of $\varepsilon$ for each of the semidefinite
programs, in the discussions of strong duality in the two previous
sections.
To argue that accurate approximate solutions to the semidefinite
programs can be obtained by the ellipsoid method, a sufficiently large
lower bounds on the value of $\varepsilon$ for which
$\D^{\circ}_{\varepsilon}$ is nonempty is needed.

For the semidefinite program for the maximum output fidelity
characterization of the completely bounded trace norm, presented in
Section~\ref{sec:max-output-fidelity}, the adjoint of the mapping
$\Xi$ is given by 
\[
\Xi^{\ast}
\begin{pmatrix}
  \lambda_0 & \cdot & \cdot & \cdot\\
  \cdot & \lambda_1 & \cdot & \cdot\\
  \cdot & \cdot & Y_0 & \cdot\\
  \cdot & \cdot & \cdot & Y_1
\end{pmatrix}
=
\begin{pmatrix}
  \lambda_0 \I_{\X} - \Psi_0^{\ast}(Y_0) & 0 & 0 & 0\\
  0 & \lambda_1 \I_{\X} - \Psi_1^{\ast}(Y_1) & 0 & 0\\
  0 & 0 & Y_0 & 0\\
  0 & 0 & 0 & Y_1
\end{pmatrix}.
\]
The operator
\[
\begin{pmatrix}
  \lambda_0 & 0 & 0 & 0\\
  0 & \lambda_1 & 0 & 0\\
  0 & 0 & \I & 0\\
  0 & 0 & 0 & \I
\end{pmatrix}
\]
for
\[
\lambda_0 = \frac{1}{2} + \norm{\Psi_0^{\ast}(\I)}_{\infty}
\qquad\text{and}\qquad
\lambda_1 = \frac{1}{2} + \norm{\Psi_1^{\ast}(\I)}_{\infty}
\]
is a specific example of a strictly dual feasible solution satisfying
\[
\Xi^{\ast}
\begin{pmatrix}
  \lambda_0 & \cdot & \cdot & \cdot\\
  \cdot & \lambda_1 & \cdot & \cdot\\
  \cdot & \cdot & \I & \cdot\\
  \cdot & \cdot & \cdot & \I
\end{pmatrix}
- \frac{1}{2}
\begin{pmatrix}
  0 & 0 & 0 & 0\\
  0 & 0 & 0 & 0\\
  0 & 0 & 0 & \I\\
  0 & 0 & \I & 0
\end{pmatrix}
\geq \frac{1}{2}
\begin{pmatrix}
  \I & 0 & 0 & 0\\
  0 & \I & 0 & 0\\
  0 & 0 & \I & 0\\
  0 & 0 & 0 & \I
\end{pmatrix}
\]
A calculation reveals that for
$H\in\Herm(\complex\oplus\complex\oplus\Z\oplus\Z)$ satisfying
\[
\norm{H}_2 \leq \frac{1}{4}\min\Bigl\{
\norm{\Psi^{\ast}_0}_{\infty}^{-1},\:
\norm{\Psi^{\ast}_1}_{\infty}^{-1},\:
1
\Bigr\}
\]
it holds that
$\norm{\Xi^{\ast}(H)}_{\infty} \leq 1/2$.
As $\Psi_0^{\ast}$ and $\Psi_1^{\ast}$ are positive, it holds that
$\norm{\Psi_0^{\ast}}_{\infty} = \norm{\Psi_0^{\ast}(\I)}_{\infty}$
and $\norm{\Psi_1^{\ast}}_{\infty} = \norm{\Psi_1^{\ast}(\I)}_{\infty}$, 
from which it follows that $\D^{\circ}_{\varepsilon}$ is nonempty for
\[
\varepsilon = 
\frac{1}{4 \bigl(1 + \norm{\Psi_0^{\ast}(\I)}_{\infty} +
\norm{\Psi_1^{\ast}(\I)}_{\infty}\bigr)}.
\]

For the semidefinite program for the completely bounded trace norm
presented in Section~\ref{sec:Choi}, based on the Choi-Jamio{\l}kowski
representation, the adjoint of the mapping $\Xi$ is given by
\[
\Xi^{\ast}
\begin{pmatrix}
  \lambda_0 & \cdot & \cdot & \cdot\\
  \cdot & \lambda_1 & \cdot & \cdot\\
  \cdot & \cdot & Y_0 & \cdot\\
  \cdot & \cdot & \cdot & Y_1
\end{pmatrix}
=
\begin{pmatrix}
  \lambda_0 \I_{\X} - \tr_{\Y}(Y_0) & 0 & 0 & 0\\
  0 & \lambda_1 \I_{\X} - \tr_{\Y}(Y_1) & 0 & 0\\
  0 & 0 & Y_0 & 0\\
  0 & 0 & 0 & Y_1
\end{pmatrix}.
\]
The operator
\[
\begin{pmatrix}
  \lambda & 0 & 0 & 0\\
  0 & \lambda & 0 & 0\\
  0 & 0 & Y & 0\\
  0 & 0 & 0 & Y
\end{pmatrix}
\]
for
\[
Y = \biggl( \frac{\norm{J(\Phi)}_{\infty}}{2} + 1 \biggr) \I_{\Y\otimes\X},
\qquad\text{and}\qquad
\lambda = 1 + 
\biggl( \frac{\norm{J(\Phi)}_{\infty}}{2} + 1 \biggr) \dim(\Y)
\]
is an example of a strictly dual feasible solution satisfying
\[
\Xi^{\ast}
\begin{pmatrix}
  \lambda & 0 & 0 & 0\\
  0 & \lambda & 0 & 0\\
  0 & 0 & Y & 0\\
  0 & 0 & 0 & Y
\end{pmatrix}
- \frac{1}{2}
\begin{pmatrix}
  0 & 0 & 0 & 0\\
  0 & 0 & 0 & 0\\
  0 & 0 & 0 & J(\Phi)\\
  0 & 0 & J(\Phi)^{\ast} & 0
\end{pmatrix}
\geq
\begin{pmatrix}
  \I & 0 & 0 & 0\\
  0 & \I & 0 & 0\\
  0 & 0 & \I & 0\\
  0 & 0 & 0 & \I
\end{pmatrix}.
\]
For
$H\in\Herm(\complex\oplus\complex\oplus(\Y\otimes\X)\oplus(\Y\otimes\X))$
satisfying
\[
\norm{H}_2 \leq \frac{1}{2\dim(\Y)}
\]
it holds that
$\norm{\Xi^{\ast}(H)}_{\infty} \leq 1$,
from which it follows that $\D^{\circ}_{\varepsilon}$ is nonempty for
\[
\varepsilon = \frac{1}{2\dim(\Y)}.
\]

In both cases, the lower bound on the value of $\varepsilon$ for which
$\D^{\circ}_{\varepsilon}$ is nonempty is polynomial in the input data
and efficiently computable.

One also requires an upper bound on the size of an optimal, or near
optimal, dual feasible solution.
For the semidefinite program based on the maximum output fidelity
characterization of the completely bounded trace norm,
every dual feasible solution is positive semidefinite, and
for approximate solutions it is sufficient to consider only those dual
feasible solutions whose trace is at most
\[
R = \norm{\Psi_0^{\ast}(\I)}_{\infty}
+ \norm{\Psi_1^{\ast}(\I)}_{\infty} + 2\dim(\Z).
\]
For the semidefinite program for the Choi-Jamio{\l}kowski
representation characterization of the completely bounded trace norm,
every dual feasible solution is again positive semidefinite, and an
optimal solution cannot have trace larger than
\[
R = 2 \norm{J(\Phi)}_{\infty} \dim(\X)\dim(\Y).
\]
The trace of every positive semidefinite operator serves as an upper
bound on that operator's Frobenius norm, which implies that the above
quantities also upper-bound the Frobenius norm of the set of dual
feasible solutions that are worthy of consideration.

As is described in detail in \cite{GrotschelLS93} for a significantly
more general setting, and summarized in \cite{Lovasz03} for the
semidefinite programming setting, the bounds $\varepsilon$
and $R$ above allow one to conclude that an algorithm running in time
polynomial in the input size and $\log(1/\delta)$ can approximate the
optimal value of the semidefinite programs discussed above to within
accuracy $\delta$.

\subsection*{Acknowledgments}

Thanks to Gus Gutoski for helpful comments on this paper.
This research was supported by Canada's NSERC and the Canadian
Institute for Advanced Research (CIFAR).

\bibliographystyle{alpha}

\begin{thebibliography}{KSV02}

\bibitem[Ali95]{Alizadeh95}
F.~Alizadeh.
\newblock Interior point methods in semidefinite programming with applications
  to combinatorial optimization.
\newblock {\em SIAM Journal on Optimization}, 5(1):13--51, 1995.

\bibitem[Bha97]{Bhatia97}
R.~Bhatia.
\newblock {\em Matrix Analysis}.
\newblock Springer, 1997.

\bibitem[dK02]{deKlerk02}
E.~de~Klerk.
\newblock {\em Aspects of Semidefinite Programming -- Interior Point Algorithms
  and Selected Applications}, volume~65 of {\em Applied Optimization}.
\newblock Kluwer Academic Publishers, Dordrecht, 2002.

\bibitem[GB09]{GrantB09}
M.~Grant and S.~Boyd.
\newblock {CVX}: {Matlab} software for disciplined convex programming.
\newblock Available from \url{http://stanford.edu/~boyd/cvx}, 2009.

\bibitem[GLS93]{GrotschelLS93}
M.~Gr\"otschel, L.~Lov\'asz, and A.~Schrijver.
\newblock {\em Geometric Algorithms and Combinatorial Optimization}.
\newblock Springer--Verlag, second corrected edition, 1993.

\bibitem[Kil12]{Killoran12}
N.~Killoran.
\newblock {\em Entanglement quantification and quantum benchmarking of optical
  communication devices}.
\newblock PhD thesis, University of Waterloo, 2012.

\bibitem[KSV02]{KitaevSV02}
A.~Kitaev, A.~Shen, and M.~Vyalyi.
\newblock {\em Classical and Quantum Computation}, volume~47 of {\em Graduate
  Studies in Mathematics}.
\newblock American Mathematical Society, 2002.

\bibitem[Lov03]{Lovasz03}
L.~Lov\'asz.
\newblock Semidefinite programs and combinatorial optimization.
\newblock {\em Recent Advances in Algorithms and Combinatorics}, 2003.

\bibitem[NC00]{NielsenC00}
M.~A. Nielsen and I.~L. Chuang.
\newblock {\em Quantum Computation and Quantum Information}.
\newblock Cambridge University Press, 2000.

\bibitem[Pau02]{Paulsen02}
V.~Paulsen.
\newblock {\em Completely Bounded Maps and Operator Algebras}.
\newblock Cambridge Studies in Advanced Mathematics. Cambridge University
  Press, 2002.

\bibitem[RW05]{RosgenW05}
B.~Rosgen and J.~Watrous.
\newblock On the hardness of distinguishing mixed-state quantum computations.
\newblock In {\em Proceedings of the 20th Annual Conference on Computational
  Complexity}, pages 344--354, 2005.

\bibitem[Wat08]{Watrous08}
J.~Watrous.
\newblock Distinguishing quantum operations having few {Kraus} operators.
\newblock {\em Quantum Information and Computation}, 8(9):819--833, 2008.

\bibitem[Wat09]{Watrous09}
J.~Watrous.
\newblock Semidefinite programs for completely bounded norms.
\newblock {\em Theory of Computing}, 5(11), 2009.

\end{thebibliography}

\end{document}